\documentclass[11pt,letterpaper]{article}

\usepackage{ramp}

\newcommand{\parencite}[1]{\cite{#1}}

\newcommand{\Ber}{\mathsf{Ber}}

\usepackage{scalerel}

\newcommand{\vAMP}{v_{\mathrm{AMP}}}
\newcommand{\vamp}{v_{\mathrm{AMP}}}

\DeclareMathOperator{\op}{\mathsf{op}}
\DeclareMathOperator{\plim}{\operatornamewithlimits{p-lim}}

\DeclareMathOperator{\PL}{PL}

\newcommand{\ind}{\boldsymbol{1}}
\DeclareMathOperator{\clip}{\mathsf{clip}}
\DeclareMathOperator{\cutoff}{\mathsf{cutoff}}
\newcommand{\sign}{\mathsf{sign}}

\author{Misha Ivkov\thanks{Stanford University. \texttt{mishai@stanford.edu}. Supported by NSF Graduate Research Fellowship.} \and Tselil Schramm\thanks{Stanford University.  \texttt{tselil@stanford.edu}. Supported by NSF CAREER award \# 2143246.}}

\title{Fast, robust approximate message passing}
\date{\today}

\begin{document}
\maketitle

\begin{abstract}
We give a fast, spectral procedure for implementing approximate-message passing (AMP) algorithms robustly.
For any quadratic optimization problem over symmetric matrices $X$ with independent subgaussian entries, and any separable AMP algorithm $\calA$, our algorithm performs a spectral pre-processing step and then mildly modifies the iterates of $\calA$.
If given the perturbed input $X + E \in \R^{n \times n}$ for any $E$ supported on a $\eps n \times \eps n$ principal minor, our algorithm outputs a solution $\hat v$ which is guaranteed to be close to the output of $\calA$ on the uncorrupted $X$, with $\|\calA(X) - \hat v\|_2 \le f(\eps) \|\calA(X)\|_2$ where $f(\eps) \to 0$ as $\eps \to 0$ depending only on $\eps$.

\end{abstract}

{ \hypersetup{hidelinks} \tableofcontents }
\thispagestyle{empty}
\clearpage
\setcounter{page}{1}

\section{Introduction}

Approximate Message Passing (AMP) is a family of algorithmic methods which generalize matrix power iteration.
Suppose we are given a symmetric matrix $X \in \R^{n \times n}$, and our goal is to maximize the quadratic form $v^\top X v$ over vectors $v$ in some constraint set $K$.
The basic AMP algorithm starts from some initialization $x^{(0)} \in \R^n$ and computes iterates $x^{(1)},x^{(2)},\ldots$ by setting $x^{(t+1)} \approx Xf(x^{(t)})$,\footnote{The $\approx$ relation hides a lower-order additive term, the ``Onsager correction,'' which depends on $x^{(t)}$. For the sake of simplicity we ignore this in the present discussion.} 
where the ``denoiser'' $f$ is a function (of the algorithm designers' choosing) from $\R \to \R$ applied coordinate-wise.
The goal of the ``powering'' action, $X x^{(t)}$, is to increase the quadratic form, while the denoiser $f$ is chosen to bring $f(x^{(t)})$ close to the constraint set $K$.

AMP algorithms are extremely popular in high-dimensional statistics.
In this context, given a prior distribution over the matrix $X$, it is often possible to optimize the design of the denoisers $f$ in such a way that AMP gives an FPTAS, in that $x^{(t)}$ obtains an $(1-\eps)$-optimal solution for $t$ large enough as a function of $\eps$.
Introduced initially as a generalization of Belief Propagation methods from statistical physics \cite{Bolt14,DMM09,BM11}, AMP algorithms are now state-of-the-art for a variety of average-case optimization problems, including compressed sensing \cite{DMM09}, sparse Principal Components Analysis (PCA) \parencite{DM14}, linear regression \parencite{DMM09,BM11,KMSSZ12}, non-negative PCA \parencite{MR15}, and more (many additional examples may be found in the surveys \parencite{Mon12,FVRS22}).
One especially notable recent application is the breakthrough work of Montanari  for optimizing the Sherrington-Kirkpatrick Hamiltonian, an average-case version of max-cut \cite{Mon21}.

One major drawback of AMP algorithms is that they are not robust.
The NP-hardness of quadratic optimization means that, obviously, one cannot hope for the optimality of AMP on average-case inputs to generalize to arbitrary inputs $X$.
But even structured perturbations can throw AMP off \cite{CZK14,RSFS19}; for example, an additive perturbation to $X$ by a rank-$1$ matrix of large norm, or planting a principal minor of uniform sign (as described in \cite{IS24}).

Our prior work addressing this issue \cite{IS24} shows that for a certain class of adversarial corruptions, AMP can be simulated robustly by polynomial-sized semidefinite programming relaxations in the ``local statistics hierarchy.''
While this result is a proof of concept that a robust version of AMP is possible, it is perhaps more interesting from a complexity-theoretic perspective than an algorithmic one:
the semidefinite programs are of size $n^{\exp(t)}$, where $t$ is the number of AMP iterations.
When AMP is an FPTAS, the algorithm of \cite{IS24} gives a robust PTAS, but the running time is too slow to feasibly implement on any computer.

In the present work, we obtain \emph{simple and fast spectral algorithms} which run in time $O(n^3)$, while not just matching but even \emph{improving} on the robustness guarantees of \cite{IS24}.
In the ``spectral algorithms from sum-of-squares analyses'' line of work (initiated in \cite{HSSS16}), our result stands out as giving a particularly dramatic reduction in running time, as well as in yielding a significantly simpler analysis.

\subsection{Setup and definitions}

We give some necessary definitions of AMP and the noise model that we consider.
\begin{definition}[AMP algorithm]\label{def:amp}
An \emph{Approximate Message Passing algorithm} is specified by a sequence of denoiser functions $\calF = f_0,f_1,f_2,\ldots$, with $f_t : \R^{t+1} \to \R$ for each $t \in \N$. 
It takes as input a symmetric $n \times n$ matrix $X$, a number of iterations $T \in \N$, and produces a sequence of iterates $x^{(0)},x^{(1)},\ldots,x^{(T)}$, with $x^{(0)} = \vec{1}$ and 
\[
x^{(t+1)} = Xf_{t}(x^{(t)},x^{(t-1)},\ldots,x^{(0)}) - \Delta_{t}(x^{(t)},x^{(t-1)},\ldots,x^{(0)}),
\]
where $f_{t}$ is applied coordinate-wise, and $\Delta_t$ is the \emph{Onsager correction term} for decreasing correlations between iterates and is fully determined by $\calF$ (see \pref{def:onsager}).
AMP algorithms often also come with a \emph{rounding} procedure which is applied to the final iterate, in order to ensure it satisfies the optimization constraints.
\end{definition}
We note that we are considering \emph{separable} AMP algorithms (where the denoisers are applied coordinate-wise) with fixed starting point $x^{(0)} = \vec{1}$. 
In full generality AMP may relax both of these criteria, but the majority of AMP analyses are compatible with these assumptions.

\begin{example}[non-negative PCA]\label{eg:nnpca}
In the non-negative principal components analysis (PCA) problem, one is given a matrix $X \in \R^{n\times n}$ and asked to maximize $v^\top X v$ over non-negative unit vectors $v \ge 0$.
The AMP algorithm which starts from $x^{(0)} = \vec{1}$ and uniformly chooses the separable denoiser $f_{s}(x^{(s)},\ldots,x^{(0)}) = f(x^{(s)})$, with $f(x) = \max(x,0)$, is an FPTAS for non-negative PCA on $X$ with i.i.d. subgaussian entries \cite{MR15}.\footnote{Technically $x^{(t)}$ may not be a unit vector nor non-negative, but AMP algorithms such as this one usually include a final ``rounding'' step---in this case, the rounding is just applying $f(x) = \max(x,0)$ followed by projection to the unit ball.}
In this case, up to the Onsager correction, AMP coincides with projected gradient ascent with ``infinite'' step size.
\end{example}

We will allow adversarially-chosen perturbations in the following model.
\begin{definition}[$\eps$-principal minor corruption]
Given matrices $X,Y \in \R^{n\times n}$, we say $Y$ is an \emph{$\eps$-principal minor corruption} of $X$ if $Y-X$ is supported on an $\eps n \times \eps n$-principal minor.
\end{definition}

A mean-$0$ random variable $\bX$ is said to be \emph{$\sigma$-subgaussian} if for each integer $k \in \N$, $\E[|\bX|^k] \le \sigma^k k^{k/2}$.
For example, a mean-$0$ Gaussian with variance $\sigma^2$ is $\sigma$-subgaussian, and a uniformly random sign $\in \{\pm 1\}$ is $1$-subgaussian.
Note that rescaling a $\sigma$-subgaussian variable $\bX$ to $C \bX$ for constant $C$ rescales the subgaussian parameter to $C\sigma$.
\subsection{Results}
Our main theorem is the following.

\begin{theorem}[Informal version of \pref{thm:main-principal}]\label{thm:main-intro}
Suppose $\calA$ is a $T$-step AMP algorithm with $O(1)$-Lipschitz or polynomial denoiser functions.
Let $X$ be a symmetric $n \times n$ matrix with i.i.d. $\frac{O(1)}{\sqrt{n}}$-subgaussian entries having mean $0$ and variance $\frac 1n$, and let $\vamp(X)$ be the output of $\calA$ on $X$.
Then there exists an algorithm which when given access to an $\eps$-principal minor corruption $Y$ produces in time $O(\eps n^3 \log n)$ a vector $\hat v(Y)$  satisfying
\[
\|\hat v(Y) - \vamp(X)\|^2 \le O(\eps \log^d \tfrac{1}{\eps}) \cdot \|\vamp(X)\|^2,
\]
with probability $1-o(1)$ over the randomness of $X$, where $d = 1$ if the denoisers are Lipschitz, and $d = k^T$ if the denoisers are degree $\le k$ polynomials.
\end{theorem}

In words, given access to an adversarially corrupted matrix $Y$, our algorithm can find a vector $\hat v(Y)$ which is close to the output of AMP on the uncorrupted matrix $X$.\footnote{Since $X$ has bounded operator norm, this implies that $\hat v(Y)$ has objective value $\hat v^\top X\hat v$ within an additive $\tilde{O}(\sqrt{\eps})$ of the objective of $\vamp(X)$.}
The result improves on that of \cite{IS24} in that it (1) runs in time $O(\eps n^3\log n)$ rather than $n^{\exp(T)}$, and (2) guarantees that $\|\calA(X) - \hat v(Y)\| \le f(\eps)\|\calA(X)\|$ for a function $f(\eps) \to_\eps 0$ which is independent of $n$ (but does depend on $T$), whereas in $\cite{IS24}$ the function $f(\eps)$ included a multiplicative factor of $\polylog(n)$, and thus was trivial unless $\eps = o(1)$.

As noted in \cite{IS24}, an equivalent result is information-theoretically impossible under the stronger corruption model in which $X-Y$ is supported on $\eps n^2$ arbitrary entries (unless $\eps = o(n^{-1/2})$).

As a direct corollary, we can robustly simulate Montanari's algorithm \cite{Mon21} for finding the ground state of the Sherrington-Kirkpatrick Hamiltonian---that is, an approximately optimal solution for Max-Cut with i.i.d. Gaussian edge weights.
\begin{corollary}[Fast, robust Sherrington Kirkpatrick]
Suppose $X$ is a symmetric matrix with entries sampled i.i.d. from $\calN(0,\frac{1}{n})$.
Then there is an algorithm which when run on an $\eps$-principal minor corruption $Y$ of $X$, with probability $1-o(1)$ produces in time $O(\eps n^3\log n)$ a unit vector $\hat v(Y) \in \{\pm 1/\sqrt{n}\}^n$ achieving objective value $\hat v(Y)^\top X \hat v(Y) \ge \mathrm{OBJ_{AMP}} - O(\sqrt{\eps \log \frac{1}{\eps}})$.
\end{corollary}
The value $\mathrm{OBJ_{AMP}}$ is the objective value achieved by Montanari's AMP algorithm; modulo a widely-believed conjecture in statistical physics, $\mathrm{OBJ_{AMP}}$ approaches $\mathrm{OPT} = \max_{v \in \{\pm 1/\sqrt{n}} v^\top X v \approx 1.52$ as $T \to \infty$.
The corollary follows from \pref{thm:main-intro} because Montanari's denoisers are Lipschitz, and the rounding scheme applied to place the final iterate in the hypercube is also Lipschitz.

\medskip
In \pref{sec:spectral}, we give a simple proof (along similar lines as the proof of \pref{thm:main-intro}) that AMP is robust to adversarial perturbations of small spectral norm.
This fact is folklore, but we feel our proof is quite simple and may be of interest.
\subsection{Experiments}

Our algorithm is fast enough that it can be easily implemented and run on a laptop.
We have run some experiments to demonstrate the utility of our method.
We consider the non-negative PCA objective described in \pref{eg:nnpca}.
In \cite{MR15}, it was shown that AMP with denoiser function $f(x) = \max(0,x)$ is an FPTAS for $\mathrm{OPT} = \max_{v \ge 0,\|v\|=1}v^\top X v = \sqrt{2}$.

In \pref{fig:exp}, we show the result for $n = 3000, \eps = 0.02$, with the adversarial corruption given by perturbing an $\eps n\times\eps n$ principal minor by sampling two independent rank $50 = \frac{5}{6}\eps n$ Wishart matrices, each normalized to have expected Frobenius norm $100$, and adding one and subtracting the other.
Without having taken pains to optimize the running time, the implementation in Python on a laptop takes less than 5 minutes.
We have plotted (1) the correlation of our algorithm's output, $\hat v(Y)$, with $\vamp(X)$, and (2) the objective value of the output for the \emph{uncorrupted} matrix $X$, $\hat v(Y)^\top X \hat v(Y)$, as a function of the number of iterations.
For comparison, we plot in \pref{fig:exp} the performance of (a) AMP on the corrupt matrix, $\vamp(Y)$, and (b) AMP on a ``naive'' spectral cleaning $\tilde{Y}$ of $Y$, given by deleting all larger-than-expected eigenvalues.
Our procedure performs much better than AMP on the corrupt input.
Empirically, the naive cleaning performance is comparable to ours, but unlike our algorithm, the naive procedure does not come with provable guarantees for arbitrary perturbations (and we suspect the naive procedure may be succeeding due to a small-$n$ effect).

\begin{figure}
\centering
\includegraphics[width=0.49\textwidth]{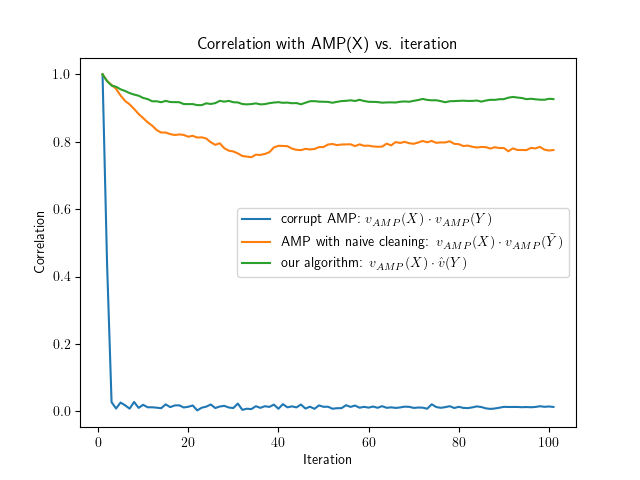}
\includegraphics[width=0.49\textwidth]{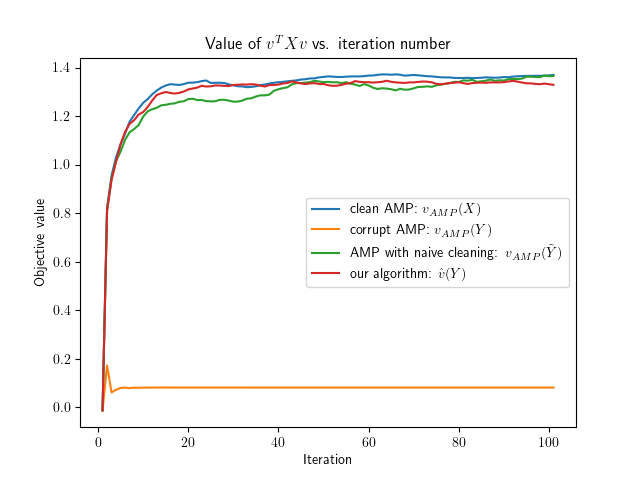}
\caption{Plot of the correlation of the vector $\hat v(Y)$ with the output of AMP on the ``clean'' matrix $X$, and of the objective value attained by $\hat v(Y)$ on the clean matrix $X$.}\label{fig:exp}
\end{figure}

\subsection{Discussion}

We give a fast spectral algorithm for simulating AMP under adversarial principal minor corruptions.
Our algorithm is an implementation of the ``spectral algorithms from sum-of-squares (SoS) analyses'' strategy introduced in \cite{HSSS16}.
We find it to be a particularly striking example of this strategy---not only was the running time reduced from $n^{\exp(T)}$ to $O(n^3)$, but also, the analysis very transparently mimics/distills that of \cite{IS24} to yield a much cleaner argument.
We draw a comparison to previous spectral-to-SoS analyses in robust statistics, most of which have been based on a ``filtering'' approach (e.g. \cite{JLST21,DKKLSS19}); in the filtering algorithms, the non-SoS analysis required significant additional tools.
Another fitting comparison is to recent works obtaining robust spectral algorithms for community recovery in the stochastic block model \cite{MRW24,DOHS23,DONS22}, where it was important to have a very fine-grained understanding of the spectrum of specific matrices.
In our case, we are able to get away with a much simpler analysis.

Though we have improved on the result in \cite{IS24} in terms of running time and the robustness-accuracy tradeoff, we differ from our prior work in one aspect: we require a description of the denoisers $\calF$ used in the AMP algorithm $\calA$, whereas the algorithm in \cite{IS24} has access only to the low-degree moments of the joint distribution over $X,\calA(X)$.
We find it unlikely that a fast algorithm could succeed without a description of $\calF$, but we pose this as a question nonetheless.

Another question is whether our error guarantees are optimal, as a function of the number of AMP iterations $T$.
In our theorem, the $\tilde{O}(\sqrt{\eps})$ hides factors that grow with the number of AMP iterations; however our experiments (\pref{fig:exp}) seem to suggest that the error stabilizes---is this a small $n$ effect?
Or perhaps an artifact of the specific perturbation from our experiments?

One clear direction for future work is making AMP robust when the input matrix $X$ has planted structure, rather than just having i.i.d. subgaussian entries.
For example, AMP has been a successful algorithm for ``spiked matrix models'' in which $X = G + \lambda uu^\top$ with $G$ a Gaussian matrix and $uu^\top$ a rank-1 spike, the goal often being to find $u$ given $X$.
In this case, it is not completely clear which noise model to study. 
In some cases (e.g. when $u$ is sparse) a principal minor corruption could simply erase the spike $uu^\top$.
However, it is an interesting question whether our techniques can be extended to this case---currently, our algorithm incorporates information about i.i.d. subgaussian variables, which makes it inappropriate for planted models (the same is true of \cite{IS24}).

Finally, it is interesting to consider alternative corruption models.
The principal minor corruption is tractable to study, and the fact that it is adversarial makes it a powerful model.
We know from \cite{IS24} that a similar result is information-theoretically impossible under the strongest sparse adversarial corruption model, in which an arbitrary subset of $\eps n^2$ entries is perturbed.
However, it would be interesting to consider alternative corruption models that more faithfully model the distribution shift one expects to see in practice, for example in the application of compressed sensing.

\subsection{Technical overview}
Though the proof of \pref{thm:main-intro} is not long, we briefly summarize the main ideas here. 
For the sake of simplicity, in this technical overview we pretend that the AMP iteration has the form $x^{(t)} = X f(x^{(t-1)})$, ignoring the Onsager correction and the dependence on more than one prior iterate.

Recall that we are given an $\eps$-principal minor corruption $Y$ of $X$.
The fact that $X$ has i.i.d. subgaussian entries of variance $\frac{1}{n}$ implies that with high probability, $\|X\|_{\op} = O(1)$.
The first step of our algorithm is a spectral procedure which removes $O(\eps n)$ rows and columns of $Y$, producing a matrix $\hat Y$ with $\|\hat Y\|_{\op} = O(1)$.
Then, we run a modified version of the AMP algorithm on the cleaned input matrix $\hat Y$, producing iterates $y^{(1)},y^{(2)},\ldots$ just as the original AMP algorithm would have, except that at each iteration we \emph{clip} the entries $y^{(t)} = \hat Y f(\clip(y^{(t-1)}))$, so that the magnitude of all entries of $\clip(y^{(t-1)})$ does not exceed the $\eps$-quantile\footnote{In the proof we choose the threshold to not exactly correspond to the $\eps$-quantile, but this choice would have also worked and is simpler for the sake of this overview.} value $O(\polylog\frac{1}{\eps})$ of the entries in a typical iterate $x^{(t-1)}$ from a \emph{clean} input matrix.

We argue that $\|y^{(t)} - x^{(t)}\| \le \tilde{O}(\sqrt{\eps}) \|x^{(t)}\|$ by induction on $t$;
In the base case, $t=0$, the iterates are identical as $x^{(t)} = \vec{1} = y^{(t)}$.
Now for $t\ge 1$, suppose that $x^{(t)}$ is the (unobserved) iterate AMP would have produced on $X$.
Then 
\begin{align}
\left\|y^{(t)} - x^{(t)}\right\|
&= \left\|\hat Y  f(\clip(y^{(t-1)})) - X f(x^{(t-1)})\right\|\nonumber \\
&\le  \left\| \hat Y(f(\clip(y^{(t-1)})) - f(x^{(t-1)}))\right\| + \left\|(\hat Y - X) f(x^{(t-1)})\right\|\nonumber \\
&\le  \left\| \hat Y\right\|_{\op}\left\|f(\clip(y^{(t-1)})) - f(x^{(t-1)})\right\| + \left\|(\hat Y - X) f(x^{(t-1)})\right\|\label{eq:terms}
\end{align}
The spectral cleaning ensures that $\|Y\|_{\op} = O(1)$.
To further bound the first term in \pref{eq:terms}, consider the illustrative case of the denoiser $f(x) = x^2$.
Then for any vectors $a,b$, $f(a) - f(b) = (a+b)\circ (a-b)$, for $\circ$ the entrywise product.
Thus we have
\begin{align}
\left\|f(\clip(y^{(t-1)})) - f(x^{(t-1)})\right\|
&= \left\|(\clip(y^{(t-1)}) + x^{(t-1)})\circ(\clip(y^{(t-1)}) - x^{(t-1)})\right\|\nonumber \\
&\le \left\|\clip(y^{(t-1)})\right\|_\infty \cdot \left\|\clip(y^{(t-1)}) - x^{(t-1)}\right\| + \left\|x^{(t-1)}\circ(\clip(y^{(t-1)}) - x^{(t-1)})\right\| \label{eq:infty}
\end{align}
The first and second terms of \pref{eq:infty} are bounded in a similar manner, we begin by explaining the first.
Because of the clipping procedure, $\|\clip(y^{(t-1)})\|_\infty = O(\polylog \frac 1\eps)$.
Further, by the triangle inequality, 
\begin{equation}
\|\clip(y^{(t-1)}) - x^{(t-1)}\| \le \|\clip(y^{(t-1)}) - \clip(x^{(t-1)})\| + \|\clip(x^{(t-1)}) - x^{(t-1)}\|.\label{eq:clip}
\end{equation}
The first term on the right of \pref{eq:clip} can be bounded by $\tilde{O}(\sqrt{\eps})\cdot \|x^{(t-1)}\|$ from the inductive hypothesis, because the $\clip$ function is $1$-Lipschitz.
The second term in \pref{eq:clip} can be bounded by $\tilde{O}(\sqrt{\eps})\cdot \|x^{(t-1)}\|$, 
because the distribution of $x^{(t-1)}$'s entries is known, and is roughly that of independent polynomials in Gaussian random variables.
To bound the second term from \pref{eq:infty}, we separate the contribution of the entries of $x^{(t-1)}$ which are bounded by $O(\polylog \frac{1}{\eps})$, to which we can apply an identical argument, and the entries which exceed this threshold, and then appeal to the fact that these integrate to a small total.
A similar argument can be used for arbitrary polynomial $f$ (for Lipschitz $f$, \pref{eq:terms} can be bounded directly and the clipping is not necessary).

To bound the second term in \pref{eq:terms}, we use the fact that $\hat Y -X$ can be written as the sum of a matrix $E$, supported on an $\eps n \times \eps n$ principal minor, and a matrix $F$ which is equal to the support of $-X$ on at most $O(\eps n)$ rows/columns---these are precisely the rows/columns of $Y$ which were removed to form $\hat Y$, but were not involved in the initial principal minor corruption.
So, $\|(\hat Y - X) f(x^{(t-1)})\| \le \|Ef(x^{(t-1)})\| + \|Ff(x^{(t-1)})\|$.
Since $E$ is supported on $\eps n$ columns, 
\[
\|E f(x^{(t-1)})\| \le \|E\|_{\op} \cdot \max_{I \subset [n],|I| = \eps n} \sum_{i\in I} f(x^{(t-1)})_i^2.
\]
Here again, because we know the order statistics of $x^{(t-1)}$, and because $f$ is required to be a well-behaved function, the maximum norm of $f(x^{(t-1)})$ when restricted to a subset of $\eps n$ coordinates is on the order of $\tilde{O}(\sqrt{\eps}) \|x^{(t-1)}\|$.
Also, since $E$ is a submatrix of $\hat Y - X$, $\|E\|_{\op} \le \|\hat Y \|_{\op} + \|X\|_{\op} \le 12$.

The matrix $F$ can be split into the part $F_1$ supported on $O(\eps n)$ columns, for which the argument is identical to the case of $E f(x^{(t-1)}$ above.
But there is also a part $F_2$ supported on $O(\eps n)$ rows.
Here, we have to take a different perspective: since $F_2$ is a restriction of $-X$ to the rows indexed by some set $T \subset [n]$ with $|T| = \eps n$, we have that $F_2 f(x^{(t-1)}) = (-X f(x^{(t-1)}))_T$, which is an $\eps n$-sparse subset of the vector $-X f(x^{(t-1)}$.
But we understand the order statistics of this vector too! 
Hence we have that $\|F_2 f(x^{(t-1)})\| = \tilde{O}(\sqrt{\eps})\|x^{(t-1)}\|$ as desired.

Putting everything together, we have that $\|y^{(t)} - x^{(t)}\| \le \tilde{O}(\sqrt{\eps}) \cdot \|x^{(t-1)}\|$.
The argument is now finished by again using our knowledge of the distribution of $x^{(t-1)}$ to conclude that $\|x^{(t-1)}\|$ and $\|x^{(t)}\|$ are within constant scalings of each other.

\medskip

Much of this analysis mirrors and simplifies the analysis in \cite{IS24}.
There, a semidefinite program is used to obtain a pseudoexpectation of a ``cleaned'' version $\hat X$ of $Y$.
The semidefinite program has formal variables for low-degree symmetric polynomials of $\hat X$.
It adds constraints to try to enforce that $\|\hat X\|_{\op} =O(1)$, that $\hat X - Y$ be supported on a principal minor (by introducing indicator variables for ``clean'' rows and columns), as well as the constraint that some symmetric vector-valued polynomials in the entries of $\hat X$ have entries which are no larger than corresponding polynomials in $X$. 

The high-level sequence of arguments mirrors those outlined in \pref{eq:terms} and the subsequent lines.
We introduce some additional structure/arguments because our spectral cleaning step (for which we design a natural-in-hindsight spectral cleaning algorithm) deletes rows and columns.
One advantage of the present argument over that in \cite{IS24} is that it is unclear how to make a semidefinite program leverage the order statistics of vector-valued polynomials, so in our prior work we crudely enforce a bound on the infinity norm of the vectors, which gives rise to $\polylog n$ factors.
Here we are able to circumvent this because we clip our iterates by hand.

\section{AMP preliminaries}

To complete \pref{def:amp} from the introduction, we must define the Onsager correction term.
\begin{definition}[Onsager correction]\label{def:onsager}
The \emph{Onsager correction term} for the AMP algorithm defined by denoisers $\calF = f_1,\ldots$ on input $X$ with iterates $x^{(0)},x^{(1)},\ldots$ is the quantity
\[
\Delta_t(v_{t},\ldots,v_0) = \sum_{j=1}^{t} B_{t,j} \cdot f_{j-1}(x^{(j-1)},\ldots,x^{(0)})
\]
where $B_{t,j} = \E_X[b_{t,j}]$ where $b_{t,j}$ is the normalized divergence of $f_t$ with respect to $x^{(j)}$:
\[ 
b_{t,j} = \frac 1n\sum_{i=1}^{n}\left.\frac{\partial f_t(x_i^{t},\ldots,u_i^{j},\ldots,,x_i^{0})}{\partial u_i^j}\right|_{u^j\rightarrow x^j}.
\]
\end{definition}

We remark that the Onsager correction is usually defined with the function $b_{t,j}$ in place of the constant $B_{t,j}$ (and in fact, generally one would estimate $B_{t,j}$ from data by computing $b_{t,j}$).
For technical reasons it is easier for us to work with $B_{t,j}$.
As was previously noted in the literature \cite[Remark 2.4]{FVRS22}, when the denoisers are well-behaved this is effectively without loss of generality because the iterates produced by using $b_{t,j}$ vs. $B_{t,j}$ are $o(1)$-close; we discuss this further in \pref{app:amp}.

\begin{definition}[Pseudo-Lipschitz Functions]
A function $\varphi: \R^t\rightarrow \R$ is called Pseudo Lipschitz of order $k$ (or $\PL(k)$) if 
\[|\varphi(x) - \varphi(y)| \le L\left(1 + \|x\|_2^{k-1} + \|y\|_2^{k-1}\right) \|x - y\|_2\]
for all $x, y\in \R^t$.
\end{definition}

Note that a function is Lipschitz exactly when it is $\PL(1)$, and a polynomial of degree $k$ lies in $\PL(k)$. By a slight abuse of notation, we will say that constants lie in $\PL(0)$.
\medskip

We will need information about the order statistics of the entries of our iterates, $x^{(t)}$.
When we run AMP with polynomial denoiser functions, each iterate $x^{(t)}$ is a symmetric (fixed by coordinate relabeling), vector-valued polynomial in the entries of $X$.
So each entry is a bounded-degree polynomial of independent subgaussian random variables. 

While the entries of $x^{(t)}$ are not independent, they are sufficiently close to independent that for simple functions $g:\R \to \R$, the average $\frac{1}{n}\sum_{i=1}^n g(x^{(t)}_i)$ concentrates fairly well around the expectation of $g$ on a polynomial of Gaussians.
The same is true when the denoiser functions are Lipschitz.
This fact is known as ``state evolution'' in the AMP literature.
In the next corollary, we state a useful consequence that will allow us to control the order statistics of our iterates.
\begin{corollary}
\torestate{
\label{cor:amp-conditioning}
	Suppose that $f: \R^{t+1} \rightarrow \R$ is $\PL(k)$ for $k\ge 0$ and $g: \R^{t+1}\rightarrow \R$ is $\PL(\ell)$ with $g(\vec 0) = 0$ and $\ell \ge 1$. 
Suppose $\vec{x} = x^{(t)}$ is an AMP iterate resulting from the application of Pseudo Lipschitz denoisers on input $X$ a symmetric matrix with i.i.d. $\frac{O(1)}{\sqrt{n}}$-subgaussian entries having mean $0$ and variance $\frac 1n$. 
Furthermore, let $C > 0$ be a constant (possibly depending on $t$).  
Then, the following hold:
	\begin{itemize}
	\item For any $r\gg \max(t,k)$,
	\[\plim_{n\rightarrow\infty} \frac 1n \sum_{i=1}^{n}f(\vec x_i)^2 \ind[g(\vec x_i)^2 > \theta] \le \frac 1{\theta^{r}}\cdot C^{2r}\cdot (3\ell r)^{\ell r}.\]
	
	\item For every $I\subseteq [n]$ with $|I| \le \eps n$,
\[\plim_{n\rightarrow\infty}\frac 1n \sum_{i\in I}f(\vec x_i)^2 \le C\eps\log^{k}\frac 1\eps.\]
	\end{itemize}
}
\end{corollary}

We prove this corollary in~\pref{app:amp}.

Sometimes we will use the phrase ``Almost-Triangle Inequality'' to refer to the inequality $(a + b)^2 \le 2a^2 + 2 b^2$.

\section{Making AMP robust to principal minor corruptions}

In this section, we prove our main theorem.

\begin{theorem}[Main Theorem]
\label{thm:main-principal}
Let $\calF$ be an AMP iteration consisting of either Lipschitz or polynomial denoiser functions. 
Suppose that $X$ is a symmetric matrix with i.i.d. entries of mean $0$, variance $\frac{1}{n}$, and subgaussian parameter $\frac{O(1)}{\sqrt{n}}$. 
Let $\vAMP(X)$ denote the output of the $T$-step AMP algorithm on input $X$, and set $d$ to be the degree of $\vAMP(X)$ as a polynomial, or $1$ if the denoisers are Lipschitz.\footnote{This aligns with the \emph{pseudo-Lipschitz} degree of $\vAMP(X)$, which functions similarly to the degree as a polynomial.}
Then, with probability $1-o(1)$ over the choice of $X$, \pref{alg:robustamp} run on any $\eps$-principal minor corruption $Y$ of $X$, produces in time $O(\eps n^3 \log n)$ a vector $\hat v(Y)$ which satisfies
\[\left\|\hat v(Y) - \vAMP(X)\right\|_2^2 \le O\left(\eps \log^{d} \frac 1\eps\right)\cdot \|\vAMP(X)\|_2^2.\]
\end{theorem}

Our algorithm consists of a pre-processing step, followed by a ``robust'' simulation of AMP:
\begin{enumerate}
	\item In the pre-processing step, we spectrally clean $Y$ by removing rows and columns to produce a matrix $\hat Y$ with $\|\hat Y\|_{\op} = O(1)$.
	\item Then, we run AMP on $\hat Y$, but with the following modification: after each iteration, we clip the iterate (coordinate-wise) to ensure all coordinates have not-too-large an absolute value.
\end{enumerate}

The following definitions will help us to describe our algorithm.

\begin{definition}
For $\eps > 0$, define $\cutoff(\eps) = \sqrt{C_T\log \frac 1\eps}$ for an appropriately large $C_T$ depending on $T$, the total number of AMP iterations.\footnote{In practice, $C_T = 16$ is a reasonable value.}
	The ``$\eps$-clip'' of $y\in \R$ is now defined to be 
\[
\clip^\eps(y) = \begin{cases}
y & |y| \le \cutoff(\eps)\\
 \sign(y)\cdot \cutoff(\eps) & |y| > \cutoff(\eps)
\end{cases}
\]
\end{definition}

\begin{definition}[Matrix restriction]
	Given a matrix $Y\in \R^{n\times n}$, $\hat{Y}$ is an $\eps$-\emph{restriction} if there exists a set $S\subseteq [n]$ with $|S| \le \eps n$ such that zeroing out the rows and columns of $Y$ with indices in $S$ yields $\hat{Y}$.
\end{definition}

Pictorially, this is as follows:
\[Y = 
\begin{bmatrix}
 	Y_{S,S} & Y_{S,\overline S}\\
 	Y_{\overline S, S} & Y_{\overline S, \overline S}
 \end{bmatrix}
\longrightarrow 
 \hat Y = \begin{bmatrix}
 \mathbf 0_{S,S} & \mathbf 0_{S,\overline S}\\
 \mathbf 0_{\overline S, S} & Y_{\overline S, \overline S}	
 \end{bmatrix}.
\]

\begin{algorithmSELF}[Robust AMP]\label{alg:robustamp}
\textbf{Input: } A symmetric $n \times n$ matrix $Y$.

\noindent \textbf{Operation: }
\vspace{-0.2cm}
\begin{enumerate}[noitemsep]
	\item Compute a restriction $\hat{Y}$ of $Y$ satisfying $\|\hat Y\|_{\op} \le 5 \cdot \E[\|X\|_{\op}]$ using~\pref{alg:spectral-cleaning}.
	\item For $t = 1,\ldots,T$, set $y^{(t)}$ to be the clipped AMP iteration
	\[y^{(t)} = \clip^{\eps}\left(\hat{Y}f_t(y^{(t-1)},\ldots,y^{(0)}) - \sum_{j=1}^{t}B_{t,j}\cdot f_{j-1}(y^{(j-1)},\ldots,y^{(0)})\right).\]
\end{enumerate}

\noindent \textbf{Output: } The vector $\hat v = y^{(T)}$.
	
\end{algorithmSELF}

\pref{thm:main-principal} is a consequence of the following two lemmas, one for each step of \pref{alg:robustamp}.

\begin{lemma}[Efficient spectral cleaning]\torestate{
\label{lem:spectral-cleaning}
Suppose $X$ is a symmetric $n \times n$ matrix with i.i.d. entries of mean zero, variance $\frac{1}{n}$, and subgaussian parameter $\frac{O(1)}{\sqrt{n}}$.
With probability $1-o(1)$ over $X$, \pref{alg:spectral-cleaning} run on any $\eps$-principal minor corruption $Y$ of $X$ with threshold value $K = 5 \E[\|X\|_{\op}]$ outputs in time $O(\eps n^3 \log n)$ a matrix $\hat Y$ which is a $4\eps$-restriction of $Y$ and satisfies $\|\hat Y\|_{\op} = O(1)$.
}
\end{lemma}

\begin{lemma}[Success of AMP on restrictions]\torestate{
\label{lem:amp-success-restriction}
Suppose $X$ is an $n\times n$ matrix with i.i.d. entries of mean zero, variance $\frac{1}{n}$ and subgaussian parameter $\frac{O(1)}{\sqrt{n}}$-subgaussian entries.
Suppose that $Y$ is an $\eps$-principal minor corruption of $X$ and $\hat{Y}$ is a $4\eps$-restriction of $Y$ with $\|\hat{Y}\|_{\op} = O(1)$. 
Then the clipped AMP iteration from \pref{alg:robustamp} on $\hat{Y}$ produces a vector $\hat{v}$ such that $\|\hat{v} - \vAMP(X)\|_2^2 \le O(\eps \log^d \frac 1\eps) \|\vAMP(X)\|_2^2$ with probability $1 - o_n(1)$ over the choice of $X$.
}	
\end{lemma}

When combined, these two lemmas immediately imply~\pref{thm:main-principal}.

\subsection{Spectral cleaning}

The goal of this section is to prove~\pref{lem:spectral-cleaning}.
Here we present the algorithm to construct $\hat Y$ which is a $4\eps$-restriction of $Y$ and has $\|\hat{Y}\|_{\op} = O(1)$.

\begin{algorithmSELF}[Spectral cleaning of principal minor corruptions]
\label{alg:spectral-cleaning}	

\textbf{Input: } A symmetric $n \times n$ matrix $Y$, and a threshold value $K \ge 0$.

\noindent\textbf{Operation: } 
\begin{enumerate}[noitemsep]	
	\item Let $\hat{Y} = Y$.

	\item While $\|\hat{Y}\|_{\op} > K$:
	\begin{enumerate}
		\item Let $v$ be the eigenvector of $\hat Y$ with eigenvalue of largest magnitude.
		\item Sample $i\in [n]$ with probability $v_i^2$.
		\item Zero out the $i$-th row and column of $\hat Y$.
	\end{enumerate}
\end{enumerate}

\noindent \textbf{Output: } Matrix $\hat{Y}$.
\end{algorithmSELF}

Note that critically we do not require that we exactly recover the corrupted rows and columns: all that matters is that we remove the indices that contribute the most to the spectral corruption. 

\begin{proof}[Proof of~\protect{\pref{lem:spectral-cleaning}}]
Certainly, the algorithm terminates since no index can be sampled more than once. 
We will show that with high probability, $O(\eps) n$ indices are removed.
The runtime bound can be deduced from noting that we have to run power iteration at most once per index removal, each run taking $O(n^2\log n)$ time. 

For convenience, let $\alpha = \E[\|X\|_{\op}]$, and recall our threshold is $K = 5\alpha$.
Since $X$ has independent subgaussian entries, with high probability $\|X\|_{\op} \le \alpha +o(1)$.
Let $Q$ denote the set of corrupted indices in $Y$. 
Furthermore, let $Y^{(0)} = Y, Y^{(1)}, \ldots, Y^{(t)}$ denote the matrix $\hat{Y}$ after each iteration of the while loop.
Similarly define $E^{(t)}$ and $Q^{(t)}$ (the non-zeroed out corrupted indices).

Note that if all indices in $Q$ are removed, the while loop will terminate (it can terminate in other instances, but this is just one stopping condition).
We show that with high probability we will reach $\|Y^{(t)}\|_{\op} \le 5\alpha$ within $4\eps n$ iterations (an thus $\hat Y$ is a $4\eps$-restriction of $Y$) using a win-win analysis: either we reach a small operator norm before removing all of $Q$ or we remove all of $Q$ (which implies the remaining matrix has norm $\le \alpha + o(1)$, because it is a principal minor of $X$).
In particular, the crux of the argument is the following:

\begin{claim}
Let $v$ be the top eigenvector of $Y^{(t)}$ and suppose $\|Y^{(t)}\|_{\op} > 5\alpha$. Then with high probability over $X$,
\[\sum_{i\in Q^{(t)}}v_i^2 \ge \frac 12 - o(1).\]
\end{claim}

Note that this claim is equivalent to saying that at each iteration of the while loop there is at least a $\frac 12$ probability of removing some index from $Q$.

\begin{proof}[Proof of Claim]
Suppose that $v^\top Y^{(t)} v > 5\alpha$ (as opposed to $v^\top Y^{(t)} v < -5\alpha$). 
Let $\tilde{v}$ be $v$ such that all indices outside of $Q^{(t)}$ are set to zero. 
Our goal is to lower bound $\|\tilde{v}\|_2^2$. 
Notice $v^\top E^{(t)} v = \tilde{v}^{\top}E^{(t)}\tilde{v}$ by definition. 
Since $v$ is the top eigenvector of $Y^{(t)}$,
\[v^\top E^{(t)} v = v^\top (Y^{(t)} - X^{(t)}) v \ge \|Y^{(t)}\|_{\op} - \|X^{(t)}\|_{\op} \ge (\|E^{(t)}\|_{\op} - \|X^{(t)}\|_{\op}) - \|X^{(t)}\|_{\op} \ge \|E^{(t)}\|_{\op} - 2\alpha - o(1)\]
where in the second to last step we used that $Y^{(t)} = X^{(t)} + E^{(t)}$ and in the last step we used that, since $X^{(t)}$ is a principal minor of $X$, $\|X^{(t)}\|_{\op} \le \|X\|_{\op} \le \alpha +o(1)$ w.h.p.

However, note that $v^\top E^{(t)} v \le \|E^{(t)}\|_{\op} \cdot \|\tilde v\|_2^2$, which implies that $\|\tilde{v}\|_2^2 \ge 1 - \frac{2\alpha + o(1)}{\|E\|_{\op}}$. 
Since $\|E^{(t)}\|_{\op} \ge \|Y^{(t)}\|_{\op} - \|X^{(t)}\|_{\op} \ge 5\alpha - (\alpha-o(1))$ by assumption, this implies that $\|\tilde{v}\|_2^2 \ge \frac 12 - o(1)$.
The proof of the case $v^\top Y v < -5\alpha$ is identical up to a change of sign.
\end{proof}

To prove that our loop terminates in $4\eps n$ steps with high probability, define the stopping time $\tau = \min\{t\ge 0: \|Y^{(t)}\|_{\op} \le 5\alpha\}$.
Now, let $I_t$ denote the indicator of whether the index removed between $Y^{(t)}$ and $Y^{(t+1)}$ was in $Q$, and note that each $I_t$ independently stochastically dominates a $B_t\sim \Ber(\frac 12 - o(1))$.
 Suppose that $\tau \ge 4\eps n$. 
 Then, it follows that
\[\sum_{j=1}^{4\eps n}B_{t-1} \le \eps n,\]
which happens with exponentially small probability (this is equivalent to asking for the probability that $\mathsf{Binomial}(4\eps n, \frac 12 - o(1)) \le \eps n$). Together, this implies that $\tau \le 4\eps n$ with high probability.\qedhere

\end{proof}

\subsection{Analysis of clipped AMP on spectrally cleaned input}

In this section we will prove~\pref{lem:amp-success-restriction}.
To begin, we examine the effect of a combination of principal minor and restriction corruptions. 
Suppose $Y$ is an $\eps$-principal minor corruption of $X$, and suppose $\hat Y$ is a $4\eps$-restriction of $Y$.
Let $S$ denote the set of rows in the support of $Y-X$, and let $T$ denote the set of rows in the support of $Y-\hat Y$. 
For simplicity, let $S' = S\setminus T$ (the set of corrupted rows which are not removed by the restriction). 
Then, the matrix evolves as follows:
\[
X = \begin{bmatrix}
 	X_{S,S} & X_{S,\overline S}\\
 	X_{\overline S, S} & X_{S,S}
 \end{bmatrix}
 \longrightarrow 
 Y = \begin{bmatrix}
 	X_{S,S} + E_{S,S} & X_{S,\overline S}\\
 	X_{\overline S, S} & X_{S,S}
 \end{bmatrix}
 \longrightarrow 
 \hat Y = \begin{bmatrix}
 	\mathbf 0_{T,T} & \mathbf 0_{T,S'} &\mathbf 0_{T,\overline{T\cup S}}\\
 	\mathbf 0_{S',T} & X_{S',S'} + E_{S',S'} & X_{S', \overline{T\cup S}}\\
 	\mathbf 0_{\overline{T\cup S}, T} & X_{\overline{T\cup S}, S'} & X_{\overline{T\cup S}, \overline{T\cup S}}
 \end{bmatrix}.
\]
In particular, if we let $E$ be the portion of the error matrix $Y-X$ which survives the restriction, and then let $F$ be the remainder in $\hat Y = X + E + F$, it follows that $F_{i,j}$ is either $-X_{i,j}$ or $0$. 
Furthermore, we will split $F$ into two sections: $F_1 \in \R^{|T|\times n}$ consisting of all entries in rows indexed by $T$ (in other words, $F_1 = -X_{T, [n]}$), and $F_2\in \R^{(n - |T|)\times |T|}$ consisting of all entries in columns indexed by $T$, except those covered by $F_1$ (in other words, $F_2 = -X_{\overline T, T})$. 
Pictorially, this can be represented via
\[
\hat Y - X = \begin{bmatrix}
 	\mathbf (-X)_{T,T} & (-X)_{T,S'} &\mathbf (-X)_{T,\overline{T\cup S}}\\
 	\mathbf (-X)_{S',T} & E_{S',S'} & \mathbf 0_{S', \overline{T\cup S}}\\
 	\mathbf (-X)_{\overline{T\cup S}, T} & \mathbf 0_{\overline{T\cup S}, S'} & \mathbf 0_{\overline{T\cup S}, \overline{T\cup S}}
 \end{bmatrix}
 =
 \begin{bmatrix}
 	F_1 & F_1 & F_1\\
 	F_2 & E_{S',S'} & \mathbf 0_{S', \overline{T\cup S}}\\
 	F_2 & \mathbf 0_{\overline{T\cup S}, S'} & \mathbf 0_{\overline{T\cup S}, \overline{T\cup S}}
 \end{bmatrix}.
 \]

As a warm-up, we show that each of these quantities is bounded in operator norm.

\begin{proposition}
	Suppose that $\|\hat Y\|_{\op} \le 5 \E[\|X\|_{\op}] =: 5\alpha$. 
	For the above definitions of $E, F_1, F_2$, we have that
	\[\|E\|_{\op} \le 6\alpha \qquad \mathrm{and}\qquad  \|F_1\|_{\op}, \|F_2\|_{\op} \le 2\alpha\]
	with high probability.
\end{proposition}

\begin{proof}
	Let us begin with $F_1$. 
	Begin by considering $F_1$, which has each entry $F_{ij}$ an independent subgaussian random variable. 
	Applying standard matrix concentration arguments (e.g. Theorem 4.6.1 in \cite{Ver18}), we have with high probability that $\|F_1\|_{\op} \le \alpha(1 + O(\sqrt {|T|/n})) \le 4\alpha$. 
	We can apply a similar argument to see that $\|F_2\|_{\op} \le 4\alpha$ as well.
	
	Now, consider $\hat{Y}_{\overline T, \overline T} = X_{\overline T, \overline T} + E$. 
	We then have that
	\[\left\|E\right\|_{\op} \le \left\|\hat{Y}_{\overline T, \overline T}\right\|_{\op} + \left\|X_{\overline T, \overline T}\right\|_{\op} \le \|\hat{Y}\|_{\op} + \|X\|_{\op} \le 6\alpha\]
	as the operator norm of a principal minor is at most that of the original matrix.
\end{proof}

With this in mind, we are ready to prove~\pref{lem:amp-success-restriction}, which we reprint here for clarity.

\restatelemma{lem:amp-success-restriction}

The proof follows from a few central claims.
The first of these shows that clipping cannot substantially change how far we are from the true AMP iteration.

\begin{proposition}[Clipping preserves error]
\label{prop:clipping-preserves}
	Define $\widetilde{y}^{(t)}$ to be the unclipped version of $y^{(t)}$ (that is, the inner expression passed to $\clip^\eps(\cdot)$). 
	Then, 
	\[\|y^{(t)} - x^{(t)}\|_2 \le \|\widetilde{y}^{(t)} - x^{(t)}\|_2 + \sqrt{\eps n}\]
	with probability $1 - o_n(1)$.
\end{proposition}

The next proposition aims to show that even though $E, F_1$, and $F_2$ have constant operator norms, their row (or column) sparsity allow for controlling their effect on AMP iterates. Here we also introduce the shorthand $f_t(x)\triangleq f_t(x^{(t-1)}, x^{(t-2)}, \ldots, x^{(0)})$.

\begin{proposition}[Block-sparse corruptions have small error]
\label{prop:block-sparse-small}
Suppose that $f_t \in \PL(d_t)$ and define $\overline d_t = \max_{j\le t} d_j$.
There exists a constant $C > 0$ (independent of $n$ and $\eps$ but possibly dependent on $t$) such that each of $\|Ef_t(x)\|_2^2, \|F_1f_t(x)\|_2^2$, and $\|F_2f_t(x)\|_2^2$ are bounded by $C\eps n \cdot \log^{\overline d_t}\frac 1\eps$ with probability $1 - o_n(1)$.
\end{proposition}

The final proposition aims to show that applying polynomials to clipped AMP iterates cannot dramatically change closeness. 
Note that this is not true in general and requires both boundedness and state evolution to hold in our case.

\begin{proposition}[Pseudo-Lipschitz functions preserve closeness of AMP iterates]
\label{prop:pl-preserves}
Suppose that $f_t\in \PL(d_t)$, and let $M = \max_{0\le i < t} \|y^{(i)} - x^{(i)}\|_2^2$.
Then, there exists a constant $C_T > 0$ (independent of $n$ and $\eps$ but dependent on $T$) such that
\[\left\|f_t(y) - f_t(x)\right\|_2^2 \le M(C_T t)^{d_t} \cdot \log^{d_t-1}\left(\tfrac 1\eps\right) + t\cdot \eps n\]
with probability $1 - o_n(1)$.
\end{proposition}

Together, these three propositions allow us to prove the lemma.

\begin{proof}[Proof of~\protect{\pref{lem:amp-success-restriction}}]
	We prove by induction on the iteration $t$. 
	Certainly, $y^{(0)} = x^{(0)} = \vec 1$ so the base case is complete. 
	Else, suppose we have proven the statement for all $k < t$. We prove for $t$.
	
	By~\pref{prop:clipping-preserves}, we have that $\|y^{(t)} - x^{(t)}\|_2^2 \le \|\widetilde y^{(t)} - x^{(t)}\|_2^2 + \eps n$, so let us handle this first term. 
	To decrease verbiage, let $\mathsf{MAX} = \max\limits_{1\le j\le t}\|f_j(y) - f_j(x)\|_2$.
	By the Triangle Inequality and the definition of the AMP iteration, we have that
\begin{align*}
\|\widetilde y^{(t)} - x^{(t)}\|_2 &= \left\|\hat Y f_t(y) - Xf_t(x)+\sum_{j=1}^{t}B_{t,j}\left(f_{j-1}(y) - f_{j-1}(x)\right)\right\|_2\\
	&\le \left\|\hat Y(f_t(y) - f_t(x) + f_t(x)) - Xf_t(x)\right\|_2 + \sum_{j=1}^{t}|B_{t,j}| \left\|f_{j-1}(y) - f_{j-1}(x)\right\|_2\\
	&\le \left\|\hat{Y}(f_t(y) - f_t(x))\right\|_2 + \left\|(\hat Y - X)f_t(x)\right\|_2 + \mathsf{MAX}\cdot \sum_{j=1}^{t}|B_{t,j}|\\
	&\le \left\|\hat{Y}\right\|_{\op}\left\|(f_t(y) - f_t(x))\right\|_2 + \left\|(\hat Y - X)f_t(x)\right\|_2 + \mathsf{MAX}\cdot \sum_{j=1}^{t}|B_{t,j}|\\
	&\le \mathsf{MAX}\cdot \left(10 + \sum_{j=1}^{t}|B_{t,j}|\right) + \|Ef_t(x)\|_2 + \|F_1 f_t(x)\|_2 + \|F_2 f_t(x)\|_2\\
	&\le C\cdot \sqrt{M(C_Tt)^{d_t} \cdot \log^{d_t-1}\left(\tfrac 1\eps\right) + t\cdot \eps n} + 3\sqrt{C\eps n\cdot \log^{\overline d_t} \left(\tfrac 1\eps\right)}
\end{align*}
where for the last inequality we applied~\pref{prop:block-sparse-small} and~\pref{prop:pl-preserves}. 
Now, the Almost-Triangle Inequality and combining with~\pref{prop:clipping-preserves} implies that
\begin{align*}
\|y^{(t)} - x^{(t)}\|_2^2 &\le 2C^2\left(M(C_Tt)^{d_t} \cdot \log^{d_t-1}\left(\tfrac 1\eps\right) + t\cdot \eps n\right) + 36C\eps n\cdot \log^{\overline d_t} \left(\tfrac 1\eps\right) + 2\eps n	\\
&= M(Ct)^{d_t} \cdot \log^{d_t - 1}\left(\tfrac 1\eps\right) + C \cdot \eps n\log^{\overline d_t} \left(\tfrac 1\eps\right).
\end{align*}
If the AMP iteration consists of Lipschitz denoisers, it follows that $d_t = 1$ for all $t$ and thus $\|y^{(t)} - x^{(t)}\|_2^2 \le (Ct)^t \cdot \eps n \log \frac 1\eps$. 
Else, notice that the power of $\log \frac 1\eps$ can be at most $t \overline d_t$, which completes the proof.

\end{proof}

We finish this section by proving the three propositions.

\begin{proof}[Proof of~\protect{\pref{prop:clipping-preserves}}]
Note that $\clip^{\eps}$ is a $1$-Lipschitz function.
So, by the triangle inequality,
\begin{align*}
\|y^{(t)} - x^{(t)}\|_2 
= \|\clip^{\eps}(\tilde{y}^{(t)}) - x^{(t)}\|_2
&\le \|\clip^{\eps}(\tilde{y}^{(t)}) - \clip^{\eps}(x^{(t)})\|_2 + \|\clip^{\eps}(x^{(t)}) - x^{(t)}\|_2\\
&\le \|\tilde{y}^{(t)} - x^{(t)}\|_2 + \|\clip^{\eps}(x^{(t)}) - x^{(t)}\|_2\\
&= \|\tilde{y}^{(t)} - x^{(t)}\|_2 + \left(\sum_{i=1}^n(x^{(t)}_i)^2 \cdot \ind\left[(x^{(t)}_i)^2 > C_T \log \frac 1\eps \right]\right)^{1/2}
\end{align*}
so it remains to bound this second quantity. 
We may apply~\pref{cor:amp-conditioning} with $f(\vec x_i) = g(\vec x_i) = x_i^{(t)}$ (which is Lipschitz) and $\theta = C_T \log \frac 1\eps$, which implies that
\[\frac 1n\sum_{i=1}^{n}(x^{(t)}_i)^2 \cdot \ind\left[(x^{(t)}_i)^2 > C_T \log \frac 1\eps\right]\le \frac{1}{(C_T \log \frac 1\eps)^{r}} \cdot (C')^{2r}\cdot (3r)^r = \left(\frac{3(C')^2\cdot r}{C_T \log \frac 1\eps}\right)^{r}.\]
By choosing $C_T \ge 3e(C')^2$ and taking $r = \log \frac 1\eps$, it follows that 
\[\frac 1n\sum_{i=1}^{n}(x^{(t)}_i)^2 \cdot \ind\left[(x^{(t)}_i)^2 > C_T \log \frac 1\eps\right] \le \eps\footnote{It may seem a bit counterintuitive that the resulting error is $\eps$, not $\eps \log \frac 1\eps$, however this is because we chose a larger-than-required $C_T$. 
By carefully optimizing this constant, we can recover the $\log \frac 1\eps$ dependence at the advantage of decreasing some constants in other parts of the overall argument.}\]
from where the conclusion follows.
\end{proof}

\begin{proof}[Proof of~\protect{\pref{prop:block-sparse-small}}]
	Let $S'$ be the indices in the support of $E$, and let $T$ be the set of indices in the row-support of $F_1$ (and column-support of $F_2$), as in the figure at the beginning of the section.
	For a given vector $v$, we will define $v_{S'}$ to be the restriction of $v$ to $S'$.
	Then, note that
	\[\|Ef_t(x)\|_2^2 = \left\|E (f_t(x))_{S'}\right\|_2^2 \le \|E\|_{\op} \cdot \left\|f_t(x)_{S'}\right\|_2^2 \le 12\left\|f_t(x)_{S'}\right\|_2^2,\]
	and similarly $\|F_2 f_t(x)\|_2^2 \le \|F_2\|_{\op}\left\|f_t(x)_{T}\right\|_2^2 \le 4\left\|f_t(x)_{T}\right\|_2^2$. 
	To handle each of these, notice that $|S'|, |T| \le 4\eps n$. 
	Therefore, we may apply~\pref{cor:amp-conditioning} to deduce that
	\[\frac 1n\left\|f_t(x)_{S'}\right\|_2^2 \le C\eps \log^{d_t} \frac 1\eps\]
	and similarly for $f_t(x)_T$. 
	This implies the boundedness of $\|Ef_t(x)\|_2^2$ and $\|F_2 f_t(x)\|_2^2$.
	
	We cannot use the same argument for $\|F_1 f_t(x)\|_2^2$ because $F_1$ is supported on all columns
	Instead, let us recall that $F_1 = -X$ on its supported rows, so $F_1 f_t(x) = (-Xf_t(x))_T$ and we are trying to bound $\|F_1f_t(x)\|_2^2 = \left\|(-Xf_t(x))_T\right\|_2^2$. 
	Using the definition of the AMP iteration, we can rewrite
	\[x^{(t)} = Xf_t(x) - \sum_{j=1}^{t}B_{t,j}f_{j-1}(x) \implies -X f_t(x) = -x^{(t)}-\sum_{j=1}^{t}B_{t,j}f_{j-1}(x).\]
	Therefore, $-Xf_t(x)\in \PL(\max_{j<t}d_j)$ and is a function of the iterates $x^{(0)}, x^{(1)}, \ldots, x^{(t)}$. 
	Once more applying~\pref{cor:amp-conditioning}, it follows that
	\[\frac 1n\|F_1f_t(x)\|_2^2 = \frac 1n \left\|(-Xf_t(x))_T\right\|_2^2 \le C\eps \log^{\max_{j<t}d_j}\frac 1\eps\]
	and we are done.
\end{proof}

\begin{proof}[Proof of~\protect{\pref{prop:pl-preserves}}]
	We begin by applying the definition of $\PL(d_t)$. 
	In particular,  combined with the Almost-Triangle Inequality we find that
\begin{align}
	\|f_t(y) - f_t(x)\|_2^2 &= \sum_{i=1}^{n} (f_t(y_i) - f_t(x_i))^2\nonumber\\
	&\le L^2 \sum_{i=1}^{n}(1 + \|y_i\|^{d_t - 1} + \|x_i\|^{d_t - 1})^2\cdot \|y_i - x_i\|^2\nonumber\\
	&\le 3L^2 \sum_{i=1}^{n}\left(1 + \|y_i\|^{2(d_t - 1)} + \|x_i\|^{2(d_t - 1)}\right) \cdot \sum_{j=0}^{t-1} (y^{(j)}_i - x^{(j)}_i)^2\nonumber\\
	&\le 3L^2 \left(1 + \max_i \|y_i\|^{2(d_t - 1)}\right)\sum_{j=0}^{t-1} \left\|y^{(j)} - x^{(j)}\right\|_2^2 + \sum_{j=0}^{t-1}\sum_{i=1}^{n}\|x_i\|^{2(d_t - 1)}\cdot (y^{(j)}_i - x^{(j)}_i)^2\nonumber\\
	&\le M\cdot 6tL^2 (C_T \cdot t\log \tfrac 1\eps)^{d_t - 1} + \sum_{j=0}^{t-1}\sum_{i=1}^{n}\|x_i\|^{2(d_t - 1)}\cdot (y^{(j)}_i - x^{(j)}_i)^2\label{eq:split-pl}.
\end{align}
The last inequality holds because for each $i$, $\|y_i\|^{2(d_t - 1)} = \left(\sum_{j=0}^{t-1}(y^{(j)}_i)^2\right)^{d_t - 1} \le (C_T\cdot t\log \frac 1\eps)^{d_t - 1}$. 
Therefore, it remains to handle the last sum.

We claim that
\[\|x_i\|^{2(d_t - 1)}\cdot (y^{(j)}_i - x^{(j)}_i)^2 \le \Biggl[(C_T \cdot t\log \tfrac 1\eps)^{d_t - 1} (y^{(j)}_i - x^{(j)}_i)^2\Biggr] + \Biggl[\|x_i\|^{2(d_t - 1)}\cdot (|x_i^{(j)}| + 2\sqrt{C_T\log \tfrac 1\eps})^2 \cdot \ind\left[\|x_i\|^2 > C_T\cdot t\log \tfrac 1\eps\right]\Biggr].\]
Indeed, 
\begin{itemize}
	\item If $\|x_i\|^2\le C_T \cdot t\log \frac 1\eps$, then certainly the left side is bounded by the first term.
	\item Else, note that $(y^{(j)}_i - x^{(j)}_i)^2 \le (|x_i^{(j)}| + 2\sqrt{C_T \log \frac 1\eps})^2$, where the absolute value protects against opposite signs. 
	Therefore, in this latter case we have that the left side is bounded by the second term.
\end{itemize}
Summing over $i\in [n]$, it follows that
\[\sum_{i=1}^{n}\|x_i\|^{2(d_t - 1)}\cdot (y^{(j)}_i - x^{(j)}_i)^2 \le M\cdot (C_T\cdot t\log \frac 1\eps)^{d_t - 1} + \sum_{i=1}^{n}\|x_i\|^{2(d_t - 1)}\cdot (|x_i^{(j)}| + 2\sqrt{C_T\log \tfrac 1\eps})^2 \cdot \ind\left[\|x_i\|^2 > C_T\cdot t\log \tfrac 1\eps\right]\]

Applying~\pref{cor:amp-conditioning} to this second term with $f(\vec x_i) = \|x_i\|^{d_t - 1} \cdot \left(|x_i^{(j)}| + 2\sqrt{C_T\log \tfrac 1\eps}\right)$, $g(\vec x_i) = \|x_i\|$ (which is Lipschitz), and $\theta = C_T \cdot t\log \frac 1\eps$, it follows that
\[\frac 1n \sum_{i=1}^{n}\|x_i\|^{2(d_t - 1)}\cdot (|x_i^{(j)}| + 2\sqrt{C_T\log \tfrac 1\eps})^2 \cdot \ind\left[\|x_i\|^2 > C_T\cdot t\log \tfrac 1\eps\right] \le \left(\frac{3(C')^2r}{C_T\cdot t\log \frac 1\eps}\right)^r \le \eps\]
by taking $r = \log\frac 1\eps$ and having $C_T\cdot t > 3e(C')^2$. Therefore, plugging this all back in to~\pref{eq:split-pl}, we have that
\begin{align*}
\|f_t(y) - f_t(x)\|_2^2 &\le M\cdot 6tL^2 (C_T \cdot t\log \tfrac 1\eps)^{d_t - 1} + \sum_{j=0}^{t-1}\sum_{i=1}^{n}\|x_i\|^{2(d_t - 1)}\cdot (y^{(j)}_i - x^{(j)}_i)^2\\
&\le M\cdot 6tL^2 (C_T \cdot t\log \tfrac 1\eps)^{d_t - 1} + \sum_{j=0}^{t-1}M\cdot (C_T\cdot t\log \frac 1\eps)^{d_t - 1} + \eps n\\
&\le M \cdot (C_T \cdot t)^{d_t} \cdot \log^{d_t - 1}\left(\tfrac 1\eps\right) + t\cdot \eps n 
\end{align*}
as desired, assuming that $C_T > 6L^2$.
\end{proof}

\section{AMP is robust to small spectral perturbations}\label{sec:spectral}

Here we argue that AMP is robust to spectral perturbations.

\begin{lemma}
Suppose that $X$ has independent entries of mean $0$, variance $\frac{1}{n}$, and subgaussian parameter $\frac{O(1)}{\sqrt{n}}$.
	Let $\calF$ be an AMP algorithm consisting of Lipschitz denoiser functions with Lipschitz constant at most $L$, and let $\vAMP(X)$ denote the output of the $T$-step AMP algorithm on input $X$, and $\vAMP(Y)$ denote the output of the same algorithm on input $Y$ for any $Y$ satisfying $\|Y-X\|_{\op} \le \eps$. 
Then there exists a universal constant $C$ such that with probability $1-o(1)$ over $X$,
	\[\frac 1n \|\vAMP(Y) - \vAMP(X)\|_2^2 \le \eps^2 \cdot C^{2T+2}\cdot ((T+1)!)^2.
\]
\end{lemma}
Since the starting iterate $x^{(0)} = \vec{1}$ and $X$ has entries of variance $\frac{1}{n}$, the scaling $\frac{1}{n}\|\vAMP(X)\|^2 \sim 1$ is of the correct order for reasonable denoisers, in which case the above implies that $\vAMP(X),\vAMP(Y)$ are $1-O(\eps^2)$-correlated.

\begin{proof}
Let us denote the iterates to be $y^{(t)}$ and $x^{(t)}$ for $Y$ and $X$, respectively, and let $E = Y-X$. 
When $t = 0$, $x^{(0)} = y^{(0)} = \vec 1$ so the statement trivially holds.
Now assuming we have shown this for all $k < t$, we will now show it $t$. 
We will use the shorthand $f_k(x) = f_k(x^{(k)},\ldots,x^{(0)})$.
We may expand the expression for $y^{(t)} - x^{(t)}$:
\begin{align*}
\left\|y^{(t)} - x^{(t)}\right\|_2 &= \left\|Yf_t(y) - Xf_t(x) + \sum_{j=1}^{t}B_{t,j}\left(f_j(y) - f_j(x)\right)\right\|_2\\
&= \left\|Y\big(f_t(y) - f_t(x)\big) + (Y-X) f_t(x) +\sum_{j=1}^{t} B_{t,j}\cdot \left(f_j(y) - f_j(x)\right)\right\|_2\\
&\le \|Y\|_{\op}\|f_t(y) - f_t(x)\| + \|Y-X\|_{\op}\|f_t(x)\| +\sum_{j=1}^{t} |B_{t,j}|\cdot \left\|f_j(y) - f_j(x)\right\|\\
\intertext{And since $\|Y\|_{\op} \le 2 \|X\|_{\op} = O(1)$ with high probability and $|B_{t,j}| = O(1)$ from the subgaussianity of $X$, and $\|Y-X\|_{\op} \le \eps$ by assumption, for a constant $C$ sufficiently large,}
&\le (Ct+C/2)\max_{k \le t}\left\|f_k(y) - f_k(x)\right\|_2 + \eps \left\|f_t(x)\right\|_2
\end{align*}
To control the first term, we invoke the Lipschitzness of $f$,
\[\|f_k(y) - f_k(x)\|_2^2 \le L\sum_{j=1}^k\|y^{(j)} - x^{(j)}\|_2^2 \le Lk\|y^{(k-1)} - x^{(k-1)}\|_2^2 \le (C^{k}(k)!)^2 \cdot \eps^2 n.\]

For the second term, we have from~\pref{cor:amp-conditioning} (applied with $\eps = 1$) that $\|f_t(x)\|_2 \le \frac{1}{2}C\sqrt n$. Combining these facts, we find that
\[\|y^{(t)} - x^{(t)}\|_2 \le (Ct+\tfrac{1}{2}C)(C^{t}t!)\cdot \eps\sqrt {n} + \tfrac{1}{2}C\eps \sqrt{n} = \eps \sqrt n\cdot C^{t + 1}(t+1)!\]
and so
\[\frac{1}{n}\|y^{(t)} - x^{(t)}\|_2^2 \le \eps^2 \cdot (C^{t+1} (t+1)!)^2.\qedhere\]
\end{proof}

\subsection*{Acknowledgments}
We thank Spencer Compton, Sam Hopkins and Andrea Montanari for helpful discussions.

\bibliographystyle{alpha}
\bibliography{main}

\newcommand{\etalchar}[1]{$^{#1}$}
\begin{thebibliography}{DKK{\etalchar{+}}19}

\bibitem[BLM12]{DBLP:conf/isit/BayatiLM12}
Mohsen Bayati, Marc Lelarge, and Andrea Montanari.
\newblock Universality in polytope phase transitions and iterative algorithms.
\newblock In {\em Proceedings of the 2012 {IEEE} International Symposium on
  Information Theory, {ISIT} 2012, Cambridge, MA, USA, July 1-6, 2012}, pages
  1643--1647. {IEEE}, 2012.

\bibitem[BM11]{BM11}
Mohsen Bayati and Andrea Montanari.
\newblock The dynamics of message passing on dense graphs, with applications to
  compressed sensing.
\newblock {\em {IEEE} Transactions on Information Theory}, 57(2):764--785,
  2011.

\bibitem[BMN20]{berthier2020state}
Raphael Berthier, Andrea Montanari, and Phan-Minh Nguyen.
\newblock State evolution for approximate message passing with non-separable
  functions.
\newblock {\em Information and Inference: A Journal of the {IMA}}, 9(1):33--79,
  2020.

\bibitem[Bol14]{Bolt14}
Erwin Bolthausen.
\newblock An iterative construction of solutions of the {TAP} equations for the
  {S}herrington--{K}irkpatrick model.
\newblock {\em Communications in Mathematical Physics}, 325(1):333--366, 2014.

\bibitem[CL20]{DBLP:journals/corr/abs-2003-10431}
Wei{-}Kuo Chen and Wai{-}Kit Lam.
\newblock Universality of approximate message passing algorithms.
\newblock {\em CoRR}, abs/2003.10431, 2020.

\bibitem[CZK14]{CZK14}
Francesco Caltagirone, Lenka Zdeborov{\'a}, and Florent Krzakala.
\newblock On convergence of approximate message passing.
\newblock In {\em 2014 {IEEE} International Symposium on Information Theory},
  pages 1812--1816. {IEEE}, 2014.

\bibitem[DdHS23]{DOHS23}
Jingqiu Ding, Tommaso d’Orsi, Yiding Hua, and David Steurer.
\newblock Reaching {K}esten-{S}tigum threshold in the stochastic block model
  under node corruptions.
\newblock In {\em The Thirty Sixth Annual Conference on Learning Theory}, pages
  4044--4071. PMLR, 2023.

\bibitem[DdNS22]{DONS22}
Jingqiu Ding, Tommaso d'Orsi, Rajai Nasser, and David Steurer.
\newblock Robust recovery for stochastic block models.
\newblock In {\em 2021 {IEEE} 62nd Annual Symposium on Foundations of Computer
  Science ({FOCS})}, pages 387--394. {IEEE}, 2022.

\bibitem[DKK{\etalchar{+}}19]{DKKLSS19}
Ilias Diakonikolas, Gautam Kamath, Daniel Kane, Jerry Li, Jacob Steinhardt, and
  Alistair Stewart.
\newblock Sever: A robust meta-algorithm for stochastic optimization.
\newblock In {\em International Conference on Machine Learning}, pages
  1596--1606. {PMLR}, 2019.

\bibitem[DM14]{DM14}
Yash Deshpande and Andrea Montanari.
\newblock Information-theoretically optimal sparse pca.
\newblock In {\em 2014 {IEEE} International Symposium on Information Theory},
  pages 2197--2201. {IEEE}, 2014.

\bibitem[DMM09]{DMM09}
David~L Donoho, Arian Maleki, and Andrea Montanari.
\newblock Message-passing algorithms for compressed sensing.
\newblock {\em Proceedings of the National Academy of Sciences},
  106(45):18914--18919, 2009.

\bibitem[FVRS22]{FVRS22}
Oliver~Y Feng, Ramji Venkataramanan, Cynthia Rush, and Richard~J Samworth.
\newblock A unifying tutorial on approximate message passing.
\newblock {\em Foundations and Trends in Machine Learning}, 15(4):335--536,
  2022.

\bibitem[HSSS16]{HSSS16}
Samuel~B Hopkins, Tselil Schramm, Jonathan Shi, and David Steurer.
\newblock Fast spectral algorithms from sum-of-squares proofs: tensor
  decomposition and planted sparse vectors.
\newblock In {\em Proceedings of the forty-eighth annual {ACM} symposium on
  Theory of Computing}, pages 178--191, 2016.

\bibitem[IS24]{IS24}
Misha Ivkov and Tselil Schramm.
\newblock Semidefinite programs simulate approximate message passing robustly.
\newblock In {\em Proceedings of the 56th Annual {ACM} Symposium on Theory of
  Computing}, pages 348--357, 2024.

\bibitem[JLST21]{JLST21}
Arun Jambulapati, Jerry Li, Tselil Schramm, and Kevin Tian.
\newblock Robust regression revisited: Acceleration and improved estimation
  rates.
\newblock {\em Advances in Neural Information Processing Systems},
  34:4475--4488, 2021.

\bibitem[JP24]{DBLP:journals/corr/abs-2404-07881}
Chris Jones and Lucas Pesenti.
\newblock Diagram analysis of iterative algorithms.
\newblock {\em CoRR}, abs/2404.07881, 2024.

\bibitem[KMS{\etalchar{+}}12]{KMSSZ12}
Florent Krzakala, Marc M{\'e}zard, Francois Sausset, Yifan Sun, and Lenka
  Zdeborov{\'a}.
\newblock Probabilistic reconstruction in compressed sensing: algorithms, phase
  diagrams, and threshold achieving matrices.
\newblock {\em Journal of Statistical Mechanics: Theory and Experiment},
  2012(08):P08009, 2012.

\bibitem[Mon12]{Mon12}
Andrea Montanari.
\newblock Graphical models concepts in compressed sensing.
\newblock {\em Compressed Sensing: Theory and Applications}, page 394, 2012.

\bibitem[Mon21]{Mon21}
Andrea Montanari.
\newblock Optimization of the {S}herrington--{K}irkpatrick {H}amiltonian.
\newblock {\em {SIAM} Journal on Computing}, (0):FOCS19--1, 2021.

\bibitem[MR15]{MR15}
Andrea Montanari and Emile Richard.
\newblock Non-negative principal component analysis: Message passing algorithms
  and sharp asymptotics.
\newblock {\em {IEEE} Transactions on Information Theory}, 62(3):1458--1484,
  2015.

\bibitem[MRW24]{MRW24}
Sidhanth Mohanty, Prasad Raghavendra, and David~X Wu.
\newblock Robust recovery for stochastic block models, simplified and
  generalized.
\newblock In {\em Proceedings of the 56th Annual {ACM} Symposium on Theory of
  Computing}, pages 367--374, 2024.

\bibitem[RSFS19]{RSFS19}
Sundeep Rangan, Philip Schniter, Alyson~K Fletcher, and Subrata Sarkar.
\newblock On the convergence of approximate message passing with arbitrary
  matrices.
\newblock {\em {IEEE} Transactions on Information Theory}, 65(9):5339--5351,
  2019.

\bibitem[Ver18]{Ver18}
Roman Vershynin.
\newblock {\em High-dimensional probability: An introduction with applications
  in data science}, volume~47.
\newblock Cambridge university press, 2018.

\end{thebibliography}

\appendix

\addcontentsline{toc}{section}{Appendices}

\section{Statistics of AMP iterate entries}\label{app:amp}

The most important theorem for us is known as \emph{state evolution}, which intuitively states that the statistics of entries iterates of an AMP iteration behave somewhat like statistics of Gaussians. 
There are two version of state evolution that will be important to us, corresponding to polynomial and Lipschitz iterations.

\begin{theorem}[Polynomial State Evolution (e.g.~\protect{\cite[Theorem 4]{DBLP:conf/isit/BayatiLM12}} or~\protect{\cite[Theorem 4.21 and Theorem 5.2]{DBLP:journals/corr/abs-2404-07881}})]
	Suppose that $x^0, x^1, \ldots, x^T$ is an AMP iteration corresponding to polynomial denoisers, with input $X$ a symmetric matrix having i.i.d $\frac {O(1)}{\sqrt n}$-subgaussian entries with mean $0$ and variance $\frac 1n$. 
	Then, for any $\PL(k)$ function $\varphi: \R^{T+1}\rightarrow \R$, 
	\[\plim_{n\rightarrow \infty} \frac 1n \sum_{i=1}^{n}\varphi(x^T_i,x^{T-1}_i,\ldots,x^0_i) = \E[\varphi(U^T, U^{T-1}, \ldots, U^0)]\]
	where $U^0,U^1,\ldots,U^T$ form an appropriate Gaussian process with covariance independent of $n$.
\end{theorem}

To prove a similar theorem in the Lipschitz setting, we note from a remark in~\cite{DBLP:journals/corr/abs-2003-10431} that~\cite[Proposition 6]{DBLP:conf/isit/BayatiLM12} can be extended to handle Lipschitz denoiser functions, which implies the same universality as for polynomial denoisers.

\begin{theorem}[Lipschitz State Evolution (e.g.~\protect{\cite[Theorem 2.3]{FVRS22}})]
\label{thm:state-evolution}
	Suppose that $x^0, x^1, \ldots, x^T$ is an AMP iteration corresponding to Lipschitz denoisers, with input $X$ a symmetric matrix having i.i.d $\frac {O(1)}{\sqrt n}$-subgaussian entries with mean $0$ and variance $\frac 1n$.  
	Then, for any $\PL(k)$ function $\varphi: \R^{T+1}\rightarrow \R$, 
	\[\plim_{n\rightarrow \infty} \frac 1n \sum_{i=1}^{n}\varphi(x^T_i,x^{T-1}_i,\ldots,x^0_i) = \E[\varphi(U^T, U^{T-1}, \ldots, U^0)]\]
	where $U^0,U^1,\ldots,U^T$ form an appropriate Gaussian process with covariance independent of $n$.
\end{theorem}

As noted by~\cite[Remark 2.4]{FVRS22}, state evolution holds for Lipschitz denoisers when we consider $B_{t,j}$ instead of $b_{t,j}$ in the AMP iteration. By adapting the version of this proof presented in~\cite[Corollary 2]{berthier2020state}, we may also substitute $B_{t,j}$ for $b_{t,j}$ when considering polynomial denoisers.

\subsection*{Consequences of state evolution for Pseudo-Lipschitz functions}

We collect some facts about Pseudo-Lipschitz functions, after which we can prove \pref{cor:amp-conditioning}, which implies the concentration we need for order statistics of the AMP iterates.

\begin{proposition}
\label{prop:prod-pl}
	Suppose that $f: \R^t\rightarrow \R$ is $\PL(a)$ with Lipschitz constant $L_1$ and $g: \R^t\rightarrow \R$ is $\PL(b)$ with Lipschitz constant $L_2$.
	 Furthermore, suppose that $f(\vec 0) = g(\vec 0) = 0$. 
	 Then, $f\cdot g\in \PL(a + b)$ with Lipschitz constant $12L_1 L_2$.
\end{proposition}

\begin{proof}
	We may write
	\[|f(x)g(x) - f(y)g(y)| \le \tfrac 12 |f(x) - f(y)|\cdot |g(x) + g(y)| + \tfrac 12|g(x) - g(y)|\cdot |f(x) + f(y)|.\]
We can see that $|f(x)| \le L_1(1 + \|x\|^{a-1})\|x\|$ and $|g(x)|\le L_2(1 + \|x\|^{b-1})\|x\|$ by applying centeredness. 
Therefore, for the first term we obtain that
\begin{align*}
|f(x) - f(y)|\cdot |g(x) + g(y)|	 &\le L_1 (1 + \|x\|^{a - 1} + \|y\|^{a-1})\|x - y\| \cdot L_2(\|x\| + \|y\| + \|x\|^{b} + \|y\|^{b})\\
&= L_1 L_2 \|x - y\| \left[(1 + \|x\|^{a - 1} + \|y\|^{a-1}) (\|x\| + \|y\| + \|x\|^{b} + \|y\|^{b})\right]
\end{align*}
It remains to show that the product of the latter two terms is at most $C(1 + \|x\|^{a + b - 1} + \|y\|^{a + b - 1})$.
Consider casing on whether $\max(\|x\|, \|y\|)\le 1$. If this is at most $1$, the above product is bounded by $12$.
 Else, suppose without loss of generality that $\|x\| \ge \|y\|$ and $\|x\| > 1$. Then, we can bound the product by $3\|x\|^{a - 1} \cdot 4\|x\|^{b} = 12\|x\|^{a + b - 1}$. 
 This implies that
\[|f(x) - f(y)|\cdot |g(x) + g(y)| \le 12 L_1 L_2 (1 + \|x\|^{a + b - 1} + \|y\|^{a + b - 1})\|x - y\|\]
and symmetrically for $|g(x) - g(y)|\cdot |f(x) + f(y)|$. 
Thus, we have shown that $f\cdot g \in \PL(a + b)$ with Lipschitz constant $12L_1 L_2$.

\end{proof}

Finally, we prove \pref{cor:amp-conditioning}.

\restatecorollary{cor:amp-conditioning}

Before the proof, note that a priori we have no control over $\ind[g(\vec x_i)^2 > \theta]$ (this is not Pseudo-Lipschitz at any degree). However, we show that we can approximate it above and below by Pseudo-Lipschitz functions and thus still reason about it.

\begin{claim}
	There exists a sequence of Lipschitz functions $f_1(x; L)$ and $f_2(x; L)$ each having Lipschitz constant $L$ such that $f_1(x; L) \le \ind[x > \theta] \le f_2(x; L)$ and as $L\rightarrow\infty$ these bounding functions converge to $\ind[x^2 > \theta]$.
\end{claim}

\begin{proof}
	Define 
	\[f_1(x; L) = \begin{cases}
 0 & \qquad x \le \theta\\
 L(x - \theta) & \qquad 	\theta < x \le \theta + \frac 1L\\
 1 & \qquad \text{otherwise}
 \end{cases}
 \qquad \text{and} \qquad
 f_2(x; L) = \begin{cases}
 0 & \qquad x \le \theta - \frac 1L\\
 L\left(x - \theta\right) + 1 & \qquad \theta - \frac 1L <	x \le \theta\\
 1 & \qquad \text{otherwise}
 \end{cases}
\]
which by definition satisfy the given constraint.
\end{proof}

This implies that state evolution holds for indicators (and we can treat them as if they are another Lipschitz function).

\begin{proof}[Proof of~\protect{\pref{cor:amp-conditioning}}]
	Let's begin with the first bullet point. By state evolution, we have that 
\[\plim_{n\rightarrow\infty} \frac 1n \sum_{i=1}^{n}f(\vec x_i)^2 \ind[g(\vec x_i)^2 > \theta] = \E_{U\sim N(0,\Sigma)}\left[f(\vec U)^2 \ind[g(\vec U)^2 > \theta]\right],\]
where $\Sigma$ is the covariance matrix of $\vec U$.
Note that for any $r \ge 1$, $ \boldsymbol{1}[x^2 > \theta] \le \frac{x^{2r}}{\theta^{r}}$ (by case analysis on whether $x^2 > \theta$).

Therefore, we find that
\[\E_{U\sim N(0,\Sigma)}\left[f(\vec U)^2 g(\vec U)^2 \ind[g(\vec U)^2 > \theta]\right] \le \frac 1{\theta^{r}}\E_U\left[(f(\vec U)g(\vec U)^r)^{2}\right].\]

By~\pref{prop:prod-pl}, we have that $f(\vec U)g(\vec U)^r \in \PL(a + br)$.

Define $h(\vec x) = f(\Sigma^{1/2}x)\cdot g(\Sigma^{1/2}x)^r$, which is also $\PL(a + br)$ and centered with Lipschitz constant at most $(12L\|\Sigma\|_{\op})^r$.
In particular, we thus have that $h(\vec x) \le (12L\|\Sigma\|_{\op})^{r} \cdot (\|x\| + \|x\|^{a + br})$ and 
\[h(\vec x)^{2} \le 2(12L\|\Sigma\|_{\op})^{2r}(\|x\|^2 + \|x\|^{2(a + br)}).\]
Now, we may compute that
\begin{align*}
\E_{g\sim N(0,I)}[h(\vec g)^{2r}] &\le 2(12L\|\Sigma\|_{\op})^{2r} \cdot \left(1 + \sum_{k_0 + k_1 + \cdots + k_t = a + br} \,\,\,\prod_{i=0}^{t}\E\left[x_i^{2k_i}\right]\right)\\
&\le 2(12L\|\Sigma\|_{\op})^{2r}\left(1 + \binom{t+a+br}{t} (2(a + br))^{a + br}\right)\\
&\le 2(12L\|\Sigma\|_{\op})^{2r} \cdot (2(a + br))^{a + br + t}\\
&\le (24L\|\Sigma\|_{\op})^{2r} \cdot (3br)^{br}
\end{align*}
where in the last two steps we used that $r\gg \max(a,t)$.\\ \\

Now, we may use this to prove the second bullet point.
We can assume without loss of generality that $f(\vec 0) = 0$: otherwise, note that $f(x)^2 \le 2(f(x) - f(0))^2 + 2f(0)^2$ so we can with a factor of $2$ loss consider centering $f$. 
Furthermore, note that we only need to consider the top $\eps$ quantile of indices $i$ to prove the statement of the lemma. 
Now, for any $\theta > 1$, consider writing
\begin{align*}
	\sum_{i=1}^{n}f(\vec x_i)^2\boldsymbol{1}[i\text{ in }\eps\text{ quantile}] &= \sum_{i=1}^{n}f(\vec x_i)^2\boldsymbol{1}[i\text{ in }\eps\text{ quantile}]\boldsymbol{1}[f(\vec x_i)^2 \le \theta] + \sum_{i=1}^{n}f(\vec x_i)^2\boldsymbol{1}[i\text{ in }\eps\text{ quantile}]\boldsymbol{1}[f(\vec x_i)^2 > \theta]\\
	&\le \eps\theta \cdot n + \sum_{i=1}^{n}f(\vec x_i)^2\boldsymbol{1}[f(\vec x_i)^2 > \theta].
\end{align*}
Therefore, dividing both sides by $n$ and taking the $\plim$ implies that
\[\plim_{n\rightarrow \infty}\frac 1n \sum_{i=1}^{n}f(\vec x_i)^2\boldsymbol{1}[i\text{ in }\eps\text{ quantile}] \le \eps \theta + \plim_{n\rightarrow \infty} \frac 1n \sum_{i=1}^{n}f(\vec x_i)^2\boldsymbol{1}[f(\vec x_i)^2 > \theta].\]

Our goal is to show that by choosing $\theta = \Theta(\log^k \frac 1\eps)$, the latter expectation is $O(\eps \log^k \frac 1\eps)$ which would complete the proof.
By the above result, we have that
\[\plim_{n\rightarrow\infty} \frac 1n \sum_{i=1}^{n}f(\vec x_i)^2 \ind[f(\vec x_i)^2 > \theta] \le \frac 1{\theta^{r}}\cdot C^{2r} \cdot (3k r)^{k r}\]
(since $f\in \PL(k)$). 
Take $r = \log \frac 1\eps$ and $\theta = 3ek C^2\cdot r^k = \Theta_\eps(\log^k \frac 1\eps)$. 
From here, we find that
\begin{align*}
\plim_{n\rightarrow\infty} \frac 1n \sum_{i=1}^{n}f(\vec x_i)^2 \ind[f(\vec x_i)^2 > \theta] &\le \frac 1{\theta^{r}}\cdot C^{2r} \cdot (3k r)^{k r}\\
&= \left(\frac{3k C^2 \cdot r^k}{3ek C^2 \cdot r^k}\right)^r\\
&\le \eps.\end{align*}
Thus, we have that
\begin{align*}
	\plim_{n\rightarrow \infty}\frac 1n \sum_{i=1}^{n}f(\vec x_i)^2\boldsymbol{1}[i\text{ in }\eps\text{ quantile}] &\le \eps \theta + \E\left[f(\vec U)^2\boldsymbol{1}[f(U)^2 > \theta]\right] \le 3ekC^2 \cdot \eps \log^k \frac 1\eps
\end{align*}
which is exactly as desired.
\end{proof}

\end{document}